\documentclass[11pt,reqno,a4paper]{amsart}
\usepackage[margin=1.05in]{geometry}
\usepackage{amssymb}
\usepackage[T1]{fontenc}
\usepackage{microtype}
\usepackage{cite}
\usepackage{url}
\usepackage[hidelinks]{hyperref}

\theoremstyle{plain}
\newtheorem{theorem}{Theorem}[section]
\newtheorem{lemma}[theorem]{Lemma}

\theoremstyle{definition}

\newcommand{\kstar}{\kappa^*}
\title[The exact growth rate of reversible pebbling on chains]
      {The Exact Growth Rate of Space-Optimal\\ Reversible Pebbling on Chains}
\author{Tetsuo Yokoyama}
\date{Preprint, \today}
\subjclass[2020]{Primary 68Q25; Secondary 68Q05, 68W40}
\keywords{Reversible computation, pebble game, time--space trade-off,
growth rate, machine-checked proof}

\begin{document}

\begin{abstract}
We determine the exact time exponent of space-optimal reversible
pebbling on chains as $1.331742379256310\ldots$.  The growth rate of
space-optimal reach exists as a limit and admits a variational formula.
The same exponent governs complete computations at minimal space,
uniformly in the chain length.
\end{abstract}

\maketitle

\section{Introduction}\label{sec:intro}

Consider the reversible pebble game on a chain \cite{Benn89}.
A move toggles (places or removes) a
pebble on node $i$ and is legal if $i=1$ or node $i-1$ carries a pebble;
at most $S$ pebbles may be on the chain at once; the play starts with an
empty chain and ends with a single pebble on node $n$.

For positive integers $n,S$, let $F(n,S)$ be the least number of moves
of such a play.  Set $F(0,S)=0$ for $S\ge0$, and set $F(n,S)=+\infty$
when $n>0$ and no such play exists, including $S=0$.  Knill
\cite{Knill95} proved, for $n,S\ge2$, the exact recursion
\begin{equation}
  F(n,S) \;=\; \min_{1\le m<n}
  \bigl[F(m,S)+F(m,S{-}1)+F(n{-}m,S{-}1)\bigr],
  \label{eq:knill}
\end{equation}
with $F(1,S)=1$ for $S\ge1$.  For $n,S\ge1$, $F(n,S)$ is finite iff
$n\le 2^{S-1}$.  Knill also gave asymptotic bounds for $F$.
At minimal space his lower bound implies
$f_p\ge n^{1.293815\ldots-o(1)}$, where $f_p=F(2^{p-1},p)$ and
$n=2^{p-1}$, while the best published upper bound there is the cost $n^{\log_2 3}$
of Bennett's strategy \cite{LiVi96a}.  The exact exponent was
left open.
For
complete computations, which start and end with an empty chain, the best
lower bound known is weaker: Kr\'a\v{l}ovi\v{c} \cite{Kral04} proved
that the time at minimal space is not $o(n\lg n)$.
Bennett's strategy is optimal in space \cite{LiTV98}; its $3^{p-1}$ moves
are the exact cost obtained by solving its recurrence \cite{LeSh90}.

\begin{theorem}[exact rate]\label{thm:main}
For each $q>0$, let $\rho_c(q)$ be the unique positive root of
$\rho^{-q}+(1{+}\rho)^{-q}=1$, let $V(q)=1/q-\log_2\rho_c(q)$,
and let $\kstar=\sup_{q>0}V(q)$.  This supremum is attained at some
$q>0$.  Moreover, the limit $c=\lim_{p\to\infty}f_p^{1/(p-1)}$
exists, and
\[
   c \;=\; 2^{1+\kstar} .
\]
\end{theorem}

\begin{theorem}[value of the constant]\label{thm:value}
We have
\[
   0.3317423792563105 \;\le\; \kstar \;\le\; 0.3317423792563106 ,
\]
so that $2.517064836893816\le c\le 2.517064836893817$.  Consequently,
writing $n=2^{p-1}$, the space-optimal reach cost satisfies
$f_p=n^{\log_2c+o(1)}$ as $p\to\infty$, with $\log_2c=1.331742379256310\ldots$\,.
\end{theorem}

Since $\log_2c<\log_2 3$, the published upper bound is not tight:
Bennett's strategy, although optimal in space, is not asymptotically
optimal in time.

The second result concerns complete computations, which start and end
with an empty chain and pebble every node at some time.  Let $T(n,S)$ be
the least number of moves of such a computation with at most $S$ pebbles.
It is finite iff $n\le 2^{S}-1$, so that the minimal space is
$S_{\min}(n)=\lceil\log_2(n+1)\rceil$.

\begin{theorem}[complete computations]\label{thm:complete}
Uniformly as $n\to\infty$,
\[
   T\bigl(n,S_{\min}(n)\bigr) \;=\; n^{\log_2c+o(1)} .
\]
\end{theorem}

Theorems~\ref{thm:main}, \ref{thm:value} and~\ref{thm:complete} have all
been verified in Lean~4 without additional hypotheses.%
\footnote{\url{https://github.com/yokoyama-lab/reversible-pebbling-exact-rate}}

\section{The counting recursion}\label{sec:prelim}

\begin{lemma}[submultiplicativity]\label{lem:sub}
$f_{p+q-1}\le f_pf_q$ for all $p,q\ge1$.  Hence
$c=\lim_pf_p^{1/(p-1)}=\inf_{m\ge1}f_{m+1}^{1/m}$ exists, and
$c\le f_{m+1}^{1/m}$ for every integer $m\ge1$.
\end{lemma}

\begin{proof}
Run a $p$-pebble reach on the chain coarsened into blocks of length
$2^{q-1}$, and implement each of its moves by a $q$-pebble reach across
one block.  Node $2^{p+q-2}$ is then reached with a peak of $p+q-1$
pebbles in $f_pf_q$ moves.  Hence $\log f_p$ is subadditive in $p-1$, and
Fekete's lemma applies.
\end{proof}

The analysis is carried out not on $F$ itself but on its normalized
forward differences.  For $S\ge1$ and $2\le n\le2^{S-1}$ we put
\[
   \delta(n,S) \;=\; \tfrac12\bigl(F(n,S)-F(n{-}1,S)\bigr) .
\]
These are positive integers, because the values of $F$ are odd and
increase in $n$.  We call the multiset
$\mathrm{row}(S)=\{\delta(n,S):2\le n\le2^{S-1}\}$ the \emph{row} at
level $S$; it has $2^{S-1}-1$ entries and is nondecreasing in $n$.  Let
\[
   m_S=\delta(2^{S-2},S)\quad(S\ge3), \qquad
   M_S=\max_n\delta(n,S)\quad(S\ge2)
\]
be its midpoint value and its maximum, with $m_2=1$ and $M_1=0$, and let
\begin{equation}
   C_S(v) \;=\; \#\{n:2\le n\le2^{S-1},\ \delta(n,S)\le v\}
   \label{eq:count}
\end{equation}
be its counting function, with the convention $C_S(v)=0$ for $v<1$.
Substituting prefix sums into \eqref{eq:knill} represents the row at
level $S\ge2$ as a merge of smaller rows,
\begin{equation}
   \mathrm{row}(S) \;=\; \{1\}\;\cup\;\mathrm{row}(S{-}1)\;\cup\;
   \bigl\{\delta(i,S)+\delta(i,S{-}1)\bigr\}_{2\le i\le2^{S-2}} .
   \label{eq:merge}
\end{equation}
We also have $M_S=m_S+M_{S-1}$ for $S\ge2$ and $m_{S-1}\le m_S$ for $S\ge3$.

For $S\ge3$, since the row is sorted, $C_S(m_S)\ge2^{S-2}-1$, and conversely
$C_S(v)\ge2^{S-2}-1$ implies $m_S\le v$.

\begin{lemma}[closed form]\label{lem:2d}
The limit $d=\lim_SM_S^{1/(S-1)}$ exists and $c=2d$.
\end{lemma}

\begin{proof}[Sketch]
Telescoping the differences gives $f_S=1+2\sum\mathrm{row}(S)$.  At least
half of the entries are at least $m_S$ and every entry is at most $M_S$,
so $2^{S-1}m_S\le f_S\le2^{S}M_S$.  Moreover $M_S=m_S+M_{S-1}$ and
$m_{S-1}\le m_S$ give $M_S\le S\,m_S$.  Under the $(S-1)$st root the
factors $2^{\pm1}$ and $S$ contribute nothing in the limit, and the two
bounds close on $2d$.
\end{proof}

\section{A dual pair of inequalities}\label{sec:dual}

\begin{lemma}[dual pair]\label{lem:dual}
For every $S\ge2$, all positive integers $u,w,v$, and every integer
$\tau$ with $0\le\tau\le v-1$,
\begin{align}
   C_S(u{+}w) &\;\ge\; 1+C_{S-1}(u{+}w)+\min\bigl(C_{S-1}(u),C_S(w)\bigr),
   \label{eq:lower}\\
   C_S(v) &\;\le\; 1+C_{S-1}(v)+\max\bigl(C_S(\tau),C_{S-1}(v{-}1{-}\tau)\bigr) .
   \label{eq:dual}
\end{align}
\end{lemma}

Both inequalities hold for every split and every threshold, and both
right-hand sides are nondecreasing in each count they contain; they are
therefore comparison principles.  A \emph{supersolution} $\Phi$ satisfies,
for each level above its initial level and each $v\ge1$, with some
admissible threshold $\tau$,
\[
   \Phi_S(v)\ge 1+\Phi_{S-1}(v)
     +\max\bigl(\Phi_S(\tau),\Phi_{S-1}(v-1-\tau)\bigr).
\]
A \emph{subsolution} $\Psi$ satisfies, at each such level and each $v\ge2$,
with some positive integers $u,w$ such that $u+w=v$,
\[
   \Psi_S(v)\le 1+\Psi_{S-1}(v)
     +\min\bigl(\Psi_{S-1}(u),\Psi_S(w)\bigr).
\]
Require also $\Phi\ge C$ and $\Psi\le C$ on the initial level and at
$v=0$, and $\Psi_S(1)\le1+\Psi_{S-1}(1)$ above the initial level.
A supersolution dominates $C$ pointwise and hence bounds $m_S$ from
below, and a subsolution is dominated by $C$ and hence bounds $m_S$ from
above.

\section{Outline of the proofs}\label{sec:outline}

\subsection{The lower bound}
For $q>0$ and $\rho>\rho_c(q)$ let
$\theta=1-\rho^{-q}-(1{+}\rho)^{-q}$, which is positive by the choice of
$\rho$, and let
\[
   \Phi_S(v) \;=\; \theta^{-1}v^{q}\rho^{qS} .
\]
At the threshold $\tau=\lfloor v/(1{+}\rho)\rfloor$ we have both
$\tau\le v/(1{+}\rho)$ and $v-1-\tau<\rho v/(1{+}\rho)$, so that for every
$v\ge1$ the two entries of the maximum in \eqref{eq:dual} are bounded,
with no error term, by $(1{+}\rho)^{-q}\Phi_S(v)$.  The right-hand side
of \eqref{eq:dual} is therefore at most $1+(1-\theta)\Phi_S(v)$, and
since $v^{q}\rho^{qS}\ge1$ for $v\ge1$, the constant $\theta^{-1}$
absorbs the additive $1$ uniformly.  Induction on $(S,v)$ from
$C_S(0)=0$ gives $C_S\le\Phi_S$, and hence
\[
   2^{S-2}-1 \;\le\; C_S(m_S) \;\le\; \theta^{-1}m_S^{q}\rho^{qS} .
\]
Taking $q$th roots and then $(S-1)$st roots, and letting
$\rho\downarrow\rho_c(q)$, we obtain
$\liminf_Sm_S^{1/(S-1)}\ge2^{V(q)}$.  Since $M_S\ge m_S$,
Lemma~\ref{lem:2d} gives $c\ge2^{1+V(q)}$ for every $q>0$, and therefore
$c\ge2^{1+\kstar}$.

\subsection{The upper bound}
The upper bound follows from an explicit subsolution of \eqref{eq:lower}.  Let
$\varepsilon>0$ and put $\lambda=\kstar+\varepsilon$.  The subsolution is
a shape $\Psi_S$ whose logarithm decays, as $v$ decreases below a front
that advances by $\lambda$ per level, at a local rate which is constant
near the front and then increases linearly with the depth below it.  The
margin available at each step of \eqref{eq:lower} is proportional to the
distance $\lambda-V(q)$ between the speed of the front and the dispersion
relation, and that distance is bounded below by $\varepsilon$ precisely
because $\kstar$ is the supremum of $V$.  For each fixed $\varepsilon>0$,
the construction yields $\log_2m_S\le\lambda S+O_\varepsilon(1)$; hence
$M_S\le K_\varepsilon\,2^{\lambda S}$ for a constant $K_\varepsilon$
independent of $S$, by $M_S=m_S+M_{S-1}$, and
letting $\varepsilon\to0$ gives $c\le2^{1+\kstar}$ by
Lemma~\ref{lem:2d}.  Together with the lower bound, this proves
Theorem~\ref{thm:main}.

\subsection{The value of the constant}
That the supremum is attained follows from the behavior of $V$ at the
two ends, together with its continuity, which confines the maximizer to a
compact interval.  The sixteen digits of Theorem~\ref{thm:value} are
obtained using Gibbs' inequality and convexity, together with certified
rational bounds for logarithms and exponentials from series remainder
estimates.  The equation defining $\rho_c(q)$ is a Kraft equality for
two codeword lengths, so the weighted arithmetic--geometric mean
inequality bounds $V(q)$ by an explicit function of two parameters in
which the implicit root no longer occurs; that function is convex in one
of them.  Certified endpoint evaluations cover a finite partition of a
bounded interval, and an analytic estimate handles the unbounded tail.

\subsection{Complete computations}
Let $H(n,S)$ be the least number of moves from the empty chain to a
configuration containing a pebble on node $n$, with no condition on the
remaining pebbles, with $H(n,S)=+\infty$ if no such play exists.
Set $H(0,S)=0$ for $S\ge0$ and $H(n,0)=+\infty$ for $n>0$.
Reversing such a play and appending it to the
original one gives a complete computation without increasing the peak
space, and cutting a complete computation at a time when node $n$ is
pebbled gives two such plays; hence $T(n,S)=2H(n,S)$.  The half-reach
satisfies, for $n,S\ge1$, the exact recursion
\begin{equation}
   H(n,S) \;=\; \min_{0\le m\le n}\bigl[F(m,S)+H(n{-}m,S{-}1)\bigr] ,
   \label{eq:half}
\end{equation}
where the inequality $\le$ is a construction and the converse follows by
cutting an optimal play at the largest leftmost pebbled position that it
attains.  Since \eqref{eq:half} is an infimal convolution of prefix sums
of nondecreasing marginal lists, the marginal list of $H$ at level $S$ is
the sorted multiset union of the marginal lists of $F$ at all levels
$1\le i\le S$.  Counting the half-reach marginals not exceeding $v$ therefore
reduces to the counting functions of Section~\ref{sec:prelim}, and both
endpoints of a minimal-space interval are of order $c^{S+o(S)}$: the
right endpoint by summing all marginals, and the left endpoint by
applying the supersolution of the lower bound to the merged list.
Monotonicity then gives the same estimate uniformly over the interval,
which is Theorem~\ref{thm:complete}.

\section*{Acknowledgements}
The author thanks Tomoo Yokoyama for discussions in the early stage of
this work.  This work was supported by JST, CREST Grant Number
JPMJCR24I4, Japan.
The Lean~4 formalization was constructed from this manuscript using
Anthropic's then-latest AI agent.  Its verification can likewise be
reproduced independently from this manuscript with an AI agent.

\bibliographystyle{plain}
\bibliography{abbrev,ref}

\end{document}